\documentclass[prd,aps,floats,preprint,preprintnumbers,amssymb,longbibliography,nofootinbib]{revtex4-2}

\usepackage[T1]{fontenc}
\usepackage[utf8]{inputenc}
\usepackage{lmodern}
\usepackage{graphicx}
\usepackage{bm}
\usepackage{physics}
\usepackage{amsthm}
\usepackage{xcolor}
\definecolor{linkblack}{RGB}{20,20,20}
\definecolor{citeblue}{RGB}{0,70,140}
\definecolor{urlblue}{RGB}{0,90,120}
\usepackage{hyperref}
\hypersetup{
  colorlinks=true,
  linkcolor=linkblack,
  citecolor=citeblue,
  urlcolor=urlblue
}

\newtheorem{theorem}{Theorem}

\newtheorem{corollary}{Corollary}

\usepackage{setspace}
\usepackage{etoolbox}

\makeatletter
\patchcmd{\@footnotetext}
  {\reset@font\footnotesize}
  {\reset@font\footnotesize\setstretch{1}}
  {}{}
\makeatother

\begin{document}

\title{The fate of Reissner--Nordstr\"om--de~Sitter black holes:
nonequilibrium discharge and evaporation}

\author{Damien~A.~Easson}
\affiliation{Department of Physics, Arizona State University, Tempe, Arizona 85287, USA}
\affiliation{Beyond Center for Fundamental Concepts in Science, Arizona State University, Tempe, Arizona 85287, USA}

\begin{abstract}
We develop an analytic semiclassical description of
Reissner--Nordstr\"om--de~Sitter (RN--dS) evaporation by combining a
spherically reduced two-dimensional dilaton gravity model with Polyakov
anomaly backreaction.  The framework captures the causal and
thermodynamic structure of the static patch and yields closed adiabatic
evolution equations for the mass and charge.  With an outward-oriented signed flux convention, the anomaly-induced
Killing-energy flux is
$\mathcal F=(N_{\rm eff}/48\pi)(\kappa_b^2-\kappa_c^2)$, while the full
mass evolution is
$\dot M=-\mathcal F+\Phi_b\dot Q$, with $\Phi_b=Q/r_b$.  We prove
analytically that along the entire sub-Nariai neutral
Schwarzschild--de~Sitter branch $\kappa_b>\kappa_c$, so that neutral
black holes lose mass monotonically.  Schwinger
pair production provides the discharge channel.  In the rapid-discharge
regime, controlled charged trajectories become effectively neutral
on a timescale short compared with the anomaly-driven Hawking mass-loss
time and then follow the neutral SdS channel toward empty de~Sitter space.
The classical lukewarm locus \(T_b=T_c\) is therefore only the nullcline
of the anomaly-induced heat flux.  Once discharge is included, the
electromagnetic work term tilts the full semiclassical vector field away
from this curve, so the lukewarm locus is not an invariant trajectory of
the coupled flow.  When sufficiently light charged species provide an operative
rapid-discharge channel, the distinguished classical loci---cold/extremal,
charged Nariai, ultracold, and lukewarm---are not semiclassical attractors
for controlled nondegenerate trajectories.  These results give an adiabatically backreacted derivation
of the RN--dS evaporation endpoint in the regime controlled by
anomaly-induced flux and rapid charge discharge, and provide the
semiclassical background for generalized-second-law monotonicity and
conservative quantum-extremal-surface/island estimates.
\end{abstract}

\maketitle
\newpage

\section{Introduction}

Black-hole thermodynamics in de~Sitter (dS) space is intrinsically
nonequilibrium.  A black hole in a static patch coexists with a
cosmological horizon, and the two horizons generally have different
surface gravities, $\kappa_b$ and $\kappa_c$, and hence different
temperatures $T_h=\kappa_h/2\pi$.  Understanding how such a two-horizon
system evolves under Hawking radiation, charge loss, and horizon motion is
therefore essential for a consistent semiclassical account of black holes
in an accelerating universe.

The neutral Schwarzschild--de~Sitter (SdS) problem has recently been clarified from several complementary perspectives.  Four-dimensional greybody-resolved calculations and near-Nariai
backreaction analyses have addressed important aspects of black-hole
evaporation in de~Sitter space \cite{Gregory:2021ozs,Aalsma:2019rpt}. In a Lorentzian
two-dimensional anomaly-induced description, neutral SdS evaporation
admits an analytic nonequilibrium flux law and evolves monotonically
toward empty de~Sitter space~\cite{Easson:2025ekn}.  A complementary Euclidean perspective~\cite{Shi:2026wnk}, has shown that
once an observer is included, the Euclidean Nariai geometry carries the
phase appropriate to black-hole nucleation.  The authors also argued that the
Bousso--Hawking anti-evaporation branch
\cite{Bousso:1997wi} is a singular-state artifact:
for horizon-smooth states, near-Nariai black holes evaporate thermally and
relax back to empty de~Sitter.  These results clarify the neutral problem,
but they do not determine the Lorentzian evolution of charged
Reissner--Nordstr\"om--de~Sitter (RN--dS) black holes, where charge loss,
electromagnetic work, and the full mass/charge phase space become essential.

In this paper we construct an analytic semiclassical model for the joint
discharge and evaporation of charged black holes in de Sitter space.  The starting point is the
spherical reduction of four-dimensional Einstein--Maxwell--$\Lambda$
gravity, supplemented by the Polyakov anomaly action
\cite{Polyakov:1981rd} for the universal two-dimensional conformal
contribution to the radial energy flux.  With positive flux defined
outward from the black-hole horizon toward the cosmological horizon, this
sector fixes the horizon-to-horizon Killing-energy flux to be proportional
to \(\kappa_b^2-\kappa_c^2\).  The charged problem then requires one
additional ingredient: when \(Q\) evolves, the mass equation
contains the electromagnetic work term, \(\dot M=-\mathcal F+\Phi_b\dot Q .\)
Together with a near-horizon Schwinger discharge law, this gives a closed
adiabatic system for \((M,Q)\).

Classical RN--dS geometry contains several celebrated loci: the
cold/extremal branch, the charged Nariai branch, the ultracold point, and
the lukewarm curve $T_b=T_c$~\cite{Ginsparg:1982rs,Romans:1992xg,Bousso:1996au,Castro:2022cuo}.
These geometries are often discussed as candidate equilibrium or endpoint
configurations, but such interpretations are based on static solution
space rather than on the time-dependent semiclassical flow.  Earlier
analyses of charged black-hole evaporation and pair production
\cite{Gibbons:1975kk,HiscockWeems1990,Gibbons:1994ff,Kim:2008xv,Cai:2014qba}
treated local fluxes or fixed backgrounds, while thermodynamic,
Euclidean, instanton, and ``shark-fin'' analyses mapped the RN--dS
parameter space and its distinguished loci using static or
quasistatic methods~\cite{Romans:1992xg,Mann:1995vb,Wang:2002as,
Castro:2022cuo,Morvan:2022aon,Aalsma:2025lcb}.  Observer-normalized first laws are essential for the local thermodynamic
interpretation of these loci, especially near Nariai; here, by contrast,
\(\kappa_b\) and \(\kappa_c\) denote the Killing surface gravities entering
the conserved anomaly-induced Killing-energy current.  Ref.~\cite{Montero:2019ekk} derived
general evolution constraints from effective energy and charge fluxes in
the context of the \textit{Festina Lente} proposal, and near-extremal
de~Sitter black holes were studied in \cite{Bhattacharjee:2025wfv} using
a Schwarzian effective theory.  These works provide important
complementary analyses, but none produced the closed anomaly-induced
horizon-to-horizon flow studied here.

The key distinction is dynamical.  Heat
capacities, fixed-temperature ensembles, and quasistatic first-law
relations do not determine the flow generated by the coupled equations for
\(M\) and \(Q\).  In particular, the classical lukewarm curve remains
important, but only as the nullcline of the Polyakov heat flux: it is the
locus where \(\kappa_b=\kappa_c\) and hence the anomaly flux vanishes.  It
is not a full-flow mass nullcline once \(\dot Q\ne0\), because the mass
equation still contains the electromagnetic work term \(\Phi_b\dot Q\).
Thus the dynamical status of the lukewarm curve must be assessed in the
full \((M,Q)\) system as we discuss below.

The local input for the discharge channel is supported by the classic
Schwinger-discharge picture of charged black holes~\cite{Gibbons:1975kk}
and by recent worldline-instanton work. Ref.~\cite{Lin:2024jug}
computed the radial pair-production profile for extremal
Reissner--Nordstr\"om black holes in asymptotically flat space and found
that Schwinger production is exponentially localized within a Compton
wavelength of the horizon.  Their
analysis shows that the dominant contribution to the discharge rate is
controlled by the near-horizon electric field, with contributions from
larger radii exponentially suppressed.  We use this as motivation for the
horizon-local discharge law in the RN--dS static patch, while embedding it
in a global two-horizon energy-balance equation.

The central result is a two-stage endpoint mechanism.  First, in the
rapid-discharge regime, Schwinger pair production drives $|Q|\to0$ on a
timescale parametrically shorter than the anomaly-induced Hawking
mass-loss time.  Second, once the trajectory reaches the neutral
neighborhood, the neutral SdS ordering takes over.  We prove in
Appendix~\ref{app:kappa-ordering} that
\[
\kappa_b(M,0)>\kappa_c(M,0),
\qquad 0<M<M_{\mathcal N},
\]
so the anomaly-induced flux drives monotonic mass loss along the neutral
branch.  Thus controlled rapid-discharge trajectories evolve toward the
empty de~Sitter endpoint rather than toward a cold, charged Nariai,
ultracold, or lukewarm remnant.  No global ordering of $\kappa_b$ and
$\kappa_c$ in the charged RN--dS family is assumed or required; indeed,
the charged family contains the lukewarm locus where the anomaly flux
vanishes.

Our analysis clarifies the relation to recent near-extremal and
lukewarm studies.  Several works have emphasized the need for a
backreacted one-loop treatment incorporating Hawking flux, electromagnetic
discharge, and observer-based energy conservation
\cite{Aalsma:2023mkz}.  The present construction supplies this missing
ingredient within a controlled spherically symmetric adiabatic model: the
Polyakov sector fixes the conserved horizon-to-horizon flux, the
Schwinger law supplies the charge-loss channel, and the resulting
phase-space flow determines which classical loci can actually be
approached dynamically.  The same time-dependent background provides the input for a
conservative island/QES estimate.  The island prescription and its role in
black-hole evaporation are now well established
\cite{Almheiri:2019psf,Penington:2019npb,Almheiri:2020cfm}, and related
questions in cosmological and de~Sitter settings have been explored in
Refs.~\cite{Hartman:2020khs,Aalsma:2021bit,Shaghoulian:2022fop,Geng:2021wcq}.
We do not attempt a microscopic derivation of the exact fine-grained Page
curve; rather, we use the evaporation dynamics to identify the background
on which the late-time no-island/island saddle competition should be
evaluated. 

The remainder of the paper is organized as follows.
Section~\ref{sec:2d-theory} sets up the spherical reduction of
Einstein--Maxwell--$\Lambda$ gravity and fixes the RN--dS conventions used
throughout.  Section~\ref{sec:fluxes} derives the anomaly-induced
horizon-to-horizon flux and introduces the full mass evolution equation,
including the electromagnetic work term.  Section~\ref{sec:thermo}
establishes the corresponding coarse-grained two-horizon entropy balance.
Section~\ref{sec:evolution-MQ} constructs the coupled $(M,Q)$ evolution
system and the corresponding phase portrait.  Section~\ref{sec:steadystates}
analyzes the cold/extremal, charged Nariai, ultracold, and lukewarm loci in
the full flow.  Section~\ref{sec:endstates} states the endpoint theorem for
the controlled rapid-discharge regime.  Section~\ref{sec:info} gives the
information-theoretic interpretation in terms of a conservative
island/QES estimate.  Appendix~\ref{app:kappa-ordering} proves the neutral
SdS surface-gravity ordering, and Appendix~\ref{app:flux-kappa-diff}
derives the conserved-current origin of the Polyakov
$(\kappa_b^2-\kappa_c^2)$ flux structure.

\section{Dimensional reduction and effective 2D theory}\label{sec:2d-theory}

We begin with Einstein--Maxwell--$\Lambda$ gravity in four dimensions,
\begin{equation}
S_4=\frac{1}{16\pi G_4}\!\int\! d^4x\sqrt{-g}\,(R-2\Lambda)
  -\frac{1}{16\pi}\!\int\! d^4x\sqrt{-g}\,F_{\mu\nu}F^{\mu\nu},
\label{eq:S4}
\end{equation}
in Gaussian units where the RN potential takes its standard form.
Imposing spherical symmetry,
$ds^2=g_{ab}(x)dx^a dx^b+r(x)^2d\Omega_2^2$,
and defining the dilaton $X=r^2$, one finds after integrating over $S^2$
(up to a total derivative) the standard spherically reduced gravity form \cite{Grumiller:2002nm,Mann:1989gh,Klosch:1995fi}:
\begin{equation}
S_2=\frac{1}{4G_4}\!\int\! d^2x\sqrt{-g}
 \Big[X R+U(X)(\nabla X)^2+2V(X)\Big],
\label{eq:S2bare}
\end{equation}
with
\begin{equation}
U(X)=\frac{1}{2X},\qquad
V(X)=1-\Lambda X-\frac{G_4Q^2}{X}.
\label{eq:UVdefs}
\end{equation}
Here \(Q\) denotes the conserved electric charge after eliminating the
two-dimensional Maxwell field at fixed charge.  With our normalization
this gives the effective charge contribution \(-G_4Q^2/X\), and the
two-dimensional action is written with a \(+2V(X)\) term.
A Weyl rescaling to the frame
\[
\tilde g_{ab}=X^{1/2}g_{ab}
\]
removes the dilaton kinetic term and gives
\begin{equation}\label{eq:Vtilde}
\widetilde V(X)=X^{-1/2}
\left(1-\Lambda X-\frac{G_4Q^2}{X}\right).
\end{equation}

In the areal-radius coordinate \(r=\sqrt X\), we define the radial primitive
entering the Schwarzschild-gauge solution by
\begin{equation}
W'(r)=V(r^2).
\end{equation}
For the potential above,
\begin{equation}
W(r)=r-\frac{\Lambda r^3}{3}+\frac{G_4Q^2}{r}.
\end{equation}
This radial primitive \(W(r)\) is the quantity entering the
four-dimensional Schwarzschild-gauge blackening function. With this convention the equation \((r\xi)'=V(r^2)=W'(r)\) directly reproduces
the four-dimensional RN--dS metric.

The static solution then has
\begin{align}
ds^2&=-\xi(r)\,dt^2+\xi(r)^{-1}\,dr^2+r^2d\Omega_2^2,\\
\xi(r;M,Q)&=\frac{W(r)-2G_4M}{r}
=1-\frac{2G_4M}{r}+\frac{G_4Q^2}{r^2}-\frac{\Lambda r^2}{3},
\label{eq:RNds}
\end{align}
which fixes the normalization used in this paper.

\paragraph*{Conventions.---}
We adopt units \(c=\hbar=1\) throughout, while keeping \(G_4\) explicit.  For generic charged nondegenerate RN--dS data, the static patch has three
positive roots
\[
0<r_-<r_b<r_c,
\]
corresponding respectively to the inner/Cauchy, black-hole/event, and
cosmological horizons.  The surface gravities are
\begin{equation}
\kappa_h=\frac12|\xi'(r_h)|,\qquad h=b,c.
\end{equation}
The cold/extremal branch is defined by $r_-=r_b$, the charged Nariai branch
by $r_b=r_c$, and the ultracold point by $r_-=r_b=r_c$.  The neutral
Nariai mass and maximum RN--dS charge scale are
\begin{equation}
M_{\mathcal N}=\frac{1}{3G_4\sqrt{\Lambda}},\qquad
Q_{\mathcal E}=\frac{1}{2\sqrt{G_4\Lambda}},
\label{eq:phase_scales}
\end{equation}
and the phase-space variables are
\begin{equation}
\mu=\frac{M}{M_{\mathcal N}},\qquad
q=\frac{Q}{Q_{\mathcal E}}.
\label{eq:muqdefs}
\end{equation}
We take $Q\ge0$ in the displayed phase portrait; the $Q<0$ half-plane is
obtained by charge conjugation.  In formulas valid with either sign of
$Q$, the Schwinger rate is written with $|Q|$ and $\operatorname{sgn}(Q)$.

For the neutral Schwarzschild--de~Sitter case ($Q=0$), it is convenient to express
the surface gravities in terms of the three real roots $(r_0<0<r_b<r_c)$ of $\xi(r)$.
Using the factorization
\[
\xi(r) = -\frac{\Lambda}{3r}\,(r-r_b)(r-r_c)(r-r_0),
\]
and the definition $\kappa_h = \tfrac12 |\xi'(r_h)|$, one obtains the closed forms
\begin{align}
\kappa_b &=
  \frac{(r_c-r_b)(2r_b+r_c)}
       {2 r_b (r_b^2+r_b r_c + r_c^2)},  \\
\kappa_c &=
  \frac{(r_c-r_b)(r_b+2r_c)}
       {2 r_c (r_b^2+r_b r_c + r_c^2)}.
\end{align}
We use these expressions repeatedly in analyzing the semiclassical flow.

\subsection{2D effective field theory perspective}

The reduction from 4D Einstein--Maxwell--$\Lambda$ gravity to
the 2D dilaton theory in Eq.~\eqref{eq:S2bare} may be viewed as
a controlled effective field theory for the spherically symmetric,
semiclassical dynamics of the RN--dS geometry.  The dilaton encodes the
areal radius of the transverse \(S^2\), while higher spherical harmonics
are integrated out.  The blackening function \eqref{eq:RNds} and dilaton
potential \(\tilde V(X)\) in Eq.~\eqref{eq:Vtilde} reproduce the exact
classical spherically symmetric sector of the four-dimensional theory.

The Polyakov anomaly term \eqref{eq:Sanom} then supplies the universal
pure-conformal contribution to the two-dimensional trace anomaly for \(N\)
conformal fields.  In this pure Polyakov sector, the associated
Killing-energy current is conserved on each static patch and depends only
on the two-dimensional metric.  Charge therefore enters the anomaly-induced
flux only through the geometry \(\xi(r;M,Q)\), and hence through the
horizon data \(\kappa_b(M,Q)\) and \(\kappa_c(M,Q)\), rather than through
a separate charge-specific anomaly coupling.

This is the controlled approximation used below.  In a more complete
spherically reduced treatment, dilaton-dependent anomaly terms,
greybody factors, and higher-derivative corrections can modify
state-dependent details and effective coefficients
\cite{Bousso:1997cg,Kummer:1998sp,Fabbri:2003vy}.  The endpoint argument
requires only that, after rapid discharge has driven the solution into an
effectively neutral neighborhood, the neutral net flux retain the sign
fixed by the SdS ordering proved in Appendix~\ref{app:kappa-ordering}.

To include backreaction from $N$ conformal fields, we add the Polyakov anomaly term \cite{Polyakov:1981rd,Christensen:1977jc,Davies:1976ei,Birrell:1982ix},
\begin{equation}
S_{\rm anom}
   =-\frac{N}{96\pi}\!\int\! d^2x\sqrt{-g}\,R\Box^{-1}R,
\label{eq:Sanom}
\end{equation}
giving $\langle T^\mu{}_{\mu}\rangle=(N/24\pi)R$.
In the adiabatic approximation, the geometry is treated at each advanced
time as an instantaneous static RN--dS patch,
\begin{equation}
ds^2=-\xi(r;M(v),Q(v))\,dv^2+2\,dv\,dr .
\label{eq:EFmetric}
\end{equation}
The corresponding instantaneous Killing vector is
\(\zeta^\mu=(\partial_v)^\mu\), equivalently \((\partial_t)^\mu\) in
Schwarzschild coordinates.  For the covariantly conserved Polyakov stress
tensor, the associated Killing-energy current
\[
J^\mu=-T^\mu{}_\nu \zeta^\nu
\]
is conserved on each fixed static patch,
\[
\nabla_\mu J^\mu=0,
\]
up to higher adiabatic corrections in the slowly evolving geometry.

In what follows we make systematic use of the anomaly-induced Killing-energy
flux derived for the neutral Schwarzschild--de~Sitter reduction in
Ref.~\cite{Easson:2025ekn}.  Because the Polyakov term depends only on the
two-dimensional metric, turning on charge introduces no new anomaly
couplings: the charged RN--dS case differs only
through its geometry, encoded in $\xi(r;M,Q)$ and the associated horizon
surface gravities, and requires no new structural ingredient beyond
replacing the neutral metric function by Eq.~\eqref{eq:RNds}.

\section{Semiclassical fluxes and coupled evolution}\label{sec:fluxes}
We work in the Unruh--de~Sitter state, regular on the future black-hole and
cosmological horizons in the appropriate chiral sectors.  We denote by
\begin{equation}
\mathcal F(M,Q)
   =\frac{N_{\rm eff}}{48\pi}
     \bigl(\kappa_b^2(M,Q)-\kappa_c^2(M,Q)\bigr)
\label{eq:fluxlaw}
\end{equation}
the outward-oriented signed Polyakov flux; positive values correspond to
energy flow from the black-hole horizon toward the cosmological horizon.  The full adiabatic mass evolution during discharge is
\begin{equation}
\dot M = -\mathcal F + \Phi_b\dot Q,
\qquad
\Phi_b=\frac{Q}{r_b},
\label{eq:Mdot}
\end{equation}
where $\Phi_b$ is the black-hole horizon potential in our adopted gauge.

The sign of the work term follows from the black-hole horizon first law.
Differentiating \(\xi(r_b;M,Q)=0\) at fixed \(\Lambda\) gives
\[
0=\xi'(r_b)\dot r_b-\frac{2G_4}{r_b}\dot M
+\frac{2G_4Q}{r_b^2}\dot Q .
\]
Using \(\xi'(r_b)=2\kappa_b\), \(S_b=\pi r_b^2/G_4\), and
\(T_b=\kappa_b/(2\pi)\), this becomes
\[
\dot M=T_b\dot S_b+\Phi_b\dot Q,
\qquad
\Phi_b=\frac{Q}{r_b}.
\]
The outward Polyakov flux removes heat from the black-hole horizon,
\(T_b\dot S_b=-\mathcal F\), giving Eq.~\eqref{eq:Mdot}.

For fixed charge this reduces to the anomaly-only contribution
\begin{equation}
\dot M_{\rm H}=-\mathcal F.
\label{eq:dotM-Hawking}
\end{equation}

For arbitrary sign of $Q$, a useful horizon-local parametrization of the
Schwinger discharge channel is
\begin{equation}
\dot Q
   =-\bar\beta_{\rm Sch}\,\ell_*\,
     \frac{e_{\rm eff}^{2}|Q|}{2\pi r_b^2}\,
     {\rm sgn}(Q)
     \exp\!\left[-\frac{\pi m_{\rm eff}^{2}r_b^2}{e_{\rm eff}|Q|}\right],
\label{eq:dotQ-Sch}
\end{equation}
for the lightest available charged species of mass $m_{\rm eff}$ and
charge $e_{\rm eff}$.  Here \(\ell_*\) is an effective matching length
associated with the near-horizon production region, while the dimensionless
positive coefficient \(\bar\beta_{\rm Sch}\) absorbs greybody factors,
species degeneracies, dimensional-reduction effects, and the matching
between the local four-dimensional pair-production rate and the net charge
current.\footnote{The factor \(\ell_*\) restores the effective length scale
that is lost when the four-dimensional near-horizon production region is
compressed into a horizon-local two-dimensional charge-current ansatz.}

Equation~\eqref{eq:dotQ-Sch} should be understood as an effective
two-dimensional horizon-local current ansatz, not as a first-principles
integration of the four-dimensional volume production rate.  The endpoint
argument uses only its robust sign,
\[
\dot Q\,{\rm sgn}(Q)<0,
\]
together with the assumed rapid-discharge hierarchy \(t_Q\ll t_M\).
For the phase portrait we display $Q\ge0$, in which case the same formula
simply gives $\dot Q<0$.  Equations~\eqref{eq:Mdot} and
\eqref{eq:dotQ-Sch} determine the coupled flow in $(M,Q)$ space.
\paragraph*{Remark on near-horizon locality.---}
The use of a near-horizon Schwinger law in Eq.~\eqref{eq:dotQ-Sch} is
supported by the worldline-instanton analysis of Lin and Shiu~\cite{Lin:2024jug},
who computed the full radial dependence of $\Gamma(r)$ for extremal
Reissner--Nordstr\"om and found that pair production is exponentially
localized within a Compton wavelength of the black-hole horizon.  Their
results show that the dominant contribution to the discharge rate is
governed by the near-horizon region and that contributions from larger
radii are negligibly small.  Incorporating this localized discharge into
the global two-horizon flux law \eqref{eq:Mdot} is a central component of
the present work.

\section{Nonequilibrium thermodynamics}\label{sec:thermo}

For the two-derivative Einstein--Maxwell--\(\Lambda\) theory considered
here, the Wald entropy reduces to the area entropy,
\[
S_h=\frac{A_h}{4G_4}=\frac{\pi r_h^2}{G_4},
\]
and \(T_h=\kappa_h/2\pi\).  For nondegenerate horizons and fixed
\(\Lambda\), the adiabatic horizon first laws are
\[
T_b\dot S_b=\dot M-\Phi_b\dot Q,
\qquad
T_c\dot S_c=-\dot M+\Phi_c\dot Q,
\qquad
\Phi_h=\frac{Q}{r_h}.
\]
Using the full semiclassical mass evolution equation \eqref{eq:Mdot}
one finds
\[
T_b\dot S_b=-\mathcal F,
\qquad
T_c\dot S_c=\mathcal F+(\Phi_c-\Phi_b)\dot Q .
\]
Hence
\[
\dot S_b+\dot S_c
=
\mathcal F\left(\frac{1}{T_c}-\frac{1}{T_b}\right)
+
\frac{(\Phi_c-\Phi_b)\dot Q}{T_c}.
\]
Since
\[
\mathcal F=\frac{N_{\rm eff}}{48\pi}(\kappa_b^2-\kappa_c^2)
\]
with positive coefficient, and since \(T_h=\kappa_h/2\pi>0\) for
nondegenerate horizons, \(\mathcal F\) has the same sign as
\(T_b-T_c\). Therefore
\[
\mathcal F\left(\frac{1}{T_c}-\frac{1}{T_b}\right)
=
\mathcal F\,\frac{T_b-T_c}{T_bT_c}\ge0 .
\]  The second term is also nonnegative during Schwinger
discharge, since \((\Phi_c-\Phi_b)\dot Q\ge0\) for either sign of \(Q\).
Thus the adiabatic two-horizon entropy balance is monotonic in the
discharging regime. This gives the coarse-grained horizon generalized-second-law balance. The relation to fine-grained entropy and the corresponding
island/QES saddle competition is discussed separately in
Sec.~\ref{sec:info}.

\paragraph*{Killing flux versus observer-normalized thermodynamics.---}
The surface gravities \(\kappa_b\) and \(\kappa_c\) appearing in
Eq.~\eqref{eq:fluxlaw} are Killing surface gravities in the normalization
used to define the conserved Polyakov Killing-energy current.  They should
not be confused with observer-normalized temperatures used in physical
static-patch first laws.  In de~Sitter black-hole thermodynamics, changing
the normalization of the timelike Killing vector corresponds to choosing a
different static observer; in particular, the Bousso--Hawking normalization
can assign finite temperatures to the Nariai throat even though the
conventional static-patch Killing surface gravities vanish there
\cite{Aalsma:2025lcb}.  The present flux law instead tracks the signed
conserved Killing-energy transfer between the two outer horizons.  A common
positive rescaling of the Killing vector rescales both outer surface
gravities and the corresponding Killing-energy variable consistently, but
it does not alter the lukewarm nullcline \(\kappa_b=\kappa_c\), the neutral
ordering \(\kappa_b>\kappa_c\), or the direction of the phase-space flow
after all quantities are expressed in the same normalization.  Thus
observer-normalized first-law analyses and the present Killing-current
evolution answer complementary questions: the former concern local
thermodynamic response and heat capacities, while the latter determines
the direction of semiclassical horizon-to-horizon energy transport.

\section{Semiclassical evolution in the $(M,Q)$ plane}
\label{sec:evolution-MQ}

For RN--dS, semiclassical evolution in
the static patch is governed by two coupled processes: anomaly-induced
Hawking radiation and Schwinger pair production of charged particles.
These give rise to a dynamical system for the black-hole mass $M(v)$ and
charge $Q(v)$ as functions of the advanced time $v$.

We adopt the geometric conventions summarized in
Sec.~\ref{sec:2d-theory}.  For a nondegenerate charged RN--dS geometry in
the interior of the physical static-patch domain, there are three positive
roots
\[
0<r_-(M,Q)<r_b(M,Q)<r_c(M,Q),
\]
corresponding to the inner, black-hole, and cosmological horizons.  The
outer horizons have surface gravities \(\kappa_b(M,Q)\) and
\(\kappa_c(M,Q)\).  The neutral Schwarzschild--de~Sitter branch is obtained
as the \(Q\to0\) boundary of this family, where only the black-hole and
cosmological horizons remain.  Along that neutral branch,
Appendix~\ref{app:kappa-ordering} proves
\[
\kappa_b(M,0)>\kappa_c(M,0)>0,
\qquad 0<M<M_{\mathcal N}.
\]
This neutral ordering will be used below after the discharge dynamics has
driven the system into an effectively neutral neighborhood.

\subsection{Anomaly-induced Hawking flux}

As in the neutral SdS case~\cite{Easson:2025ekn}, integrating out $N$
conformal matter fields and including the Polyakov action yields a
conserved Killing flux in the Unruh--de~Sitter state.  The coefficient of
the two-dimensional Polyakov anomaly is $N$, while the effective
coefficient entering the four-dimensional matched flux law is denoted
$N_{\rm eff}$.  With the outward-oriented signed convention of
Eq.~\eqref{eq:fluxlaw}, the Polyakov contribution is $\mathcal F$.  In the
pure Polyakov sector, charge modifies this flux only through the geometry
$\xi(r;M,Q)$ and hence through the horizon data $\kappa_b(M,Q)$ and
$\kappa_c(M,Q)$.  The anomaly-only mass change at fixed charge is
\[
\dot M_{\rm H}=-\mathcal F .
\]
When discharge is present, the total mass evolution is instead the full
energy-balance relation \eqref{eq:Mdot}.

By the neutral SdS ordering proved in Appendix~\ref{app:kappa-ordering},
the anomaly flux is positive along the entire sub-Nariai neutral branch, leading to
\[
\dot M_{\rm H}(M,0)=-\mathcal F(M,0)<0,
\qquad 0<M<M_{\mathcal N},
\]
so neutral SdS black holes lose mass monotonically.

For $Q\neq0$, however, the sign of
$\kappa_b^2-\kappa_c^2$ need not be fixed.  It vanishes on the classical
lukewarm locus, and on the side where $\kappa_b<\kappa_c$ the
outward-oriented Polyakov flux is negative, so the anomaly-only
contribution to the mass change satisfies
\[
\dot M_{\rm H}>0 .
\]
This corresponds to a temporary inward Killing-energy current from the
cosmological horizon toward the black-hole horizon in the effective
two-horizon balance.  In the full system, Eq.~\eqref{eq:Mdot} shows that
this positive anomaly-only mass contribution competes with the
electromagnetic work term.  During discharge, $\Phi_b\dot Q\le0$, while
the charge equation still drives $|Q|$ downward.  Thus an initially
negative Polyakov flux does not obstruct the endpoint mechanism in the
rapid-discharge regime: the trajectory is driven across the lukewarm
nullcline into an effectively neutral neighborhood, where the neutral
ordering, together with continuity near \(Q=0\), fixes the subsequent
positive sign of the Polyakov flux and hence the remaining mass loss.
\subsection{Schwinger discharge of the black-hole charge}

The electric field near the black-hole horizon sources Schwinger pair
production of charged particles~\cite{Schwinger:1951nm}.  For a species of mass $m$ and charge $e$
with $|eQ|\gg1$, the local pair creation rate per unit four-volume is
approximately
\begin{equation}
  \Gamma_{\rm Sch}(M,Q)
  \;\sim\;
  \frac{(e E_b)^2}{4\pi^3}\,
  \exp\!\left[-\,\frac{\pi m^2}{e E_b}\right],
  \label{eq:Gamma-Sch}
\end{equation}
where $E_b=|Q|/r_b^2$ in the conventions of Eq.~\eqref{eq:RNds}.  Motivated by the near-horizon localization of the Schwinger process, we
parametrize the net horizon-local discharge current by the sign-safe
effective law \eqref{eq:dotQ-Sch}.  The dimensionless coefficient
\(\bar\beta_{\rm Sch}\) absorbs the radial profile of the production region,
greybody factors, species degeneracies, dimensional-reduction effects, and
the matching between the local four-dimensional pair-production rate and
the net charge current; the associated microscopic matching length is denoted
\(\ell_*\).  Crucially,
\begin{equation}
  \dot Q_{\rm Sch}(M,Q)\,{\rm sgn}(Q)<0,
\end{equation}
so Schwinger discharge always decreases the magnitude of the charge,
\begin{equation}
  \frac{d}{dv}|Q(v)|<0,
\end{equation}
until the black hole becomes effectively neutral.

In the physically relevant regime with at least one light charged species
obeying $eE_b\gtrsim m^2$ over a substantial portion of the evolution, the
exponential in \eqref{eq:dotQ-Sch} is not strongly suppressed, and the
associated discharge timescale is parametrically shorter than the Hawking
mass-loss timescale.  We return to this hierarchy of timescales below.

\subsection{Dimensionless dynamical system}

It is convenient to introduce dimensionless variables adapted to the RN--dS
``shark-fin'' parameter space.  We use the de~Sitter radius
$L=\sqrt{3/\Lambda}$ as the time scale and the phase-space scales in
Eq.~\eqref{eq:phase_scales}:
\begin{equation}
  \mu=\frac{M}{M_{\mathcal N}},\qquad
  q=\frac{Q}{Q_{\mathcal E}},\qquad
  \tau=\frac{v}{L}.
\end{equation}
The admissible nondegenerate static-patch geometries form a bounded domain
\(\mathcal D\) in the \((\mu,q)\) plane.

In these variables the full semiclassical flow is
\begin{align}
  \frac{d\mu}{d\tau}
  &=F(\mu,q)
   \equiv \frac{L}{M_{\mathcal N}}
      \left[-\mathcal F(M,Q)+\Phi_b(M,Q)\dot Q(M,Q)\right],
   \label{eq:dyn-mu}
  \\
  \frac{dq}{d\tau}
  &=G(\mu,q)
   \equiv \frac{L}{Q_{\mathcal E}}\dot Q(M,Q),
   \label{eq:dyn-q}
\end{align}
with $M=M_{\mathcal N}\mu$, $Q=Q_{\mathcal E}q$, and $\dot Q$ given by
Eq.~\eqref{eq:dotQ-Sch}.  Here the electromagnetic work term
$\Phi_b\dot Q$ is present in the mass equation.  The above also makes clear that the
classical lukewarm locus $\kappa_b=\kappa_c$ is only the anomaly-flux
nullcline $\mathcal F=0$; it is not the full mass nullcline when
$\dot Q\ne0$.

The sign information used below is
\begin{equation}
G(\mu,q)\,{\rm sgn}(q)<0,
\qquad
\Phi_b\dot Q\le0,
\qquad
F(\mu,0)<0\quad(0<\mu<1),
\label{eq:sign-structure}
\end{equation}
where the last inequality is the neutral SdS result.  No global sign is
assumed for the anomaly term away from the neutral line.

\subsection{Phase portrait}

The sign structure \eqref{eq:sign-structure} determines the robust part of
the qualitative phase portrait in the physical domain $\mathcal{D}$.  For
$q>0$ one has $G<0$, while for $q<0$ one has $G>0$; hence the charge
component of the vector field always points toward the neutral line.  Along
$q=0$ the charge is frozen and the mass decreases monotonically by the
neutral SdS theorem.  Away from the neutral line the anomaly term can change
sign across the lukewarm locus, but the work term $\Phi_b\dot Q$ is always
nonpositive during discharge, and in the rapid-discharge regime the charge
motion dominates before the mass changes appreciably.

Thus trajectories in the regime \(t_Q\ll t_M\) are first driven toward an
effectively neutral neighborhood of \(q=0\) and then follow the neutral
channel toward lower \(\mu\).  In the presence of at
least one light charged species with unsuppressed Schwinger production,
the characteristic discharge timescale
\begin{equation}
  t_Q \;\sim\; \frac{|Q|}{|\dot Q_{\rm Sch}|}
\end{equation}
is parametrically shorter than the Hawking mass-loss timescale
\begin{equation}
  t_M \;\sim\; \frac{M}{|\dot M_{\rm H}|},
\end{equation}
so the trajectory first experiences a rapid neutralization burst,
$|q|\to0$, followed by a slower evolution along the neutral SdS channel
$\mu\to0$.

The essential point is that rapid discharge turns the late-time problem
into the neutral SdS problem: once \(Q\) is macroscopically negligible,
the neutral ordering proved in Appendix~\ref{app:kappa-ordering} fixes the
sign of the remaining anomaly-driven Hawking evolution.

We next use this phase-portrait to identify the possible endpoints of the semiclassical evolution and show that, under
realistic assumptions about the charged matter content, neither lukewarm
nor charged Nariai configurations arise as late-time attractors.

\begin{figure}[t]
    \centering
    \includegraphics[width=0.95\columnwidth]{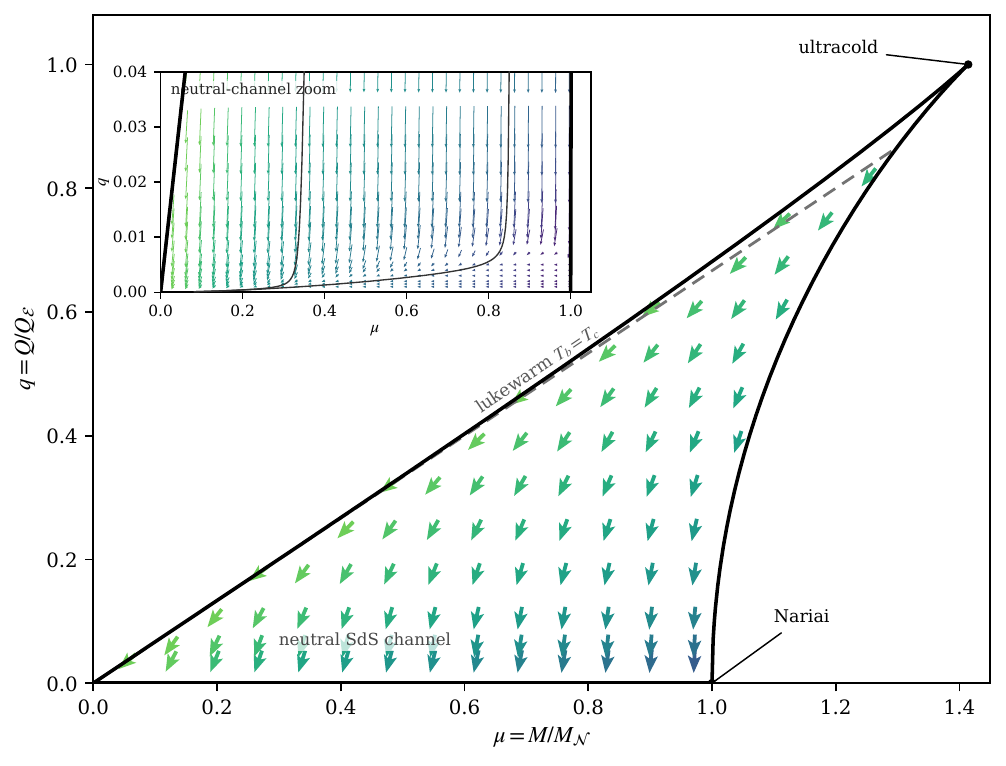}
  \caption{
Semiclassical RN--dS phase portrait in dimensionless variables
$(\mu,q)=(M/M_{\mathcal N},Q/Q_{\mathcal E})$.  The solid black curve is
the physical RN--dS double-root boundary: the cold/extremal branch and the
charged Nariai branch meet at the ultracold point, while the neutral Nariai
point lies at $(\mu,q)=(1,0)$.  The dashed gray line is the classical
lukewarm locus $T_b=T_c$, shown only as a reference curve.
The arrows show the local vector field computed from the full semiclassical
evolution equations, including the electromagnetic work term
$\dot M=-\mathcal F+\Phi_b\dot Q$, together with the Schwinger discharge
law for $\dot Q$.  Colors encode
$\log_{10}\sqrt{(d\mu/d\tau)^2+(dq/d\tau)^2}$, rescaled for visual
contrast.  The inset zooms into the near-neutral region and shows the
turnover into the neutral SdS channel; the thin curves are integrated
trajectories of the same full dynamical system.
}
 \label{fig:RNDS-sharkfin}
\end{figure}

\subsection{Quantitative hierarchy of discharge and evaporation timescales}
\label{sec:timescales}

The phase portrait in Fig.~\ref{fig:RNDS-sharkfin} is governed by the
sign structure~\eqref{eq:sign-structure}, while the detailed shape of the
trajectories is controlled by the \emph{relative} sizes of the mass-loss
and charge-loss rates.  In this
subsection we provide a quantitative estimate of the corresponding
timescales, showing that Schwinger discharge is parametrically faster than
Hawking mass loss for macroscopic black holes in the unsuppressed-discharge
regime.

\paragraph*{Hawking mass-loss timescale.---}
The anomaly-induced flux law~\eqref{eq:fluxlaw} gives
\[
\dot{M}_{\rm H}
   = -\,\frac{N_{\rm eff}}{48\pi}\,
      \bigl[\kappa_b^2(M,Q)-\kappa_c^2(M,Q)\bigr].
\]
Along the neutral branch ($Q=0$) we prove in
Appendix~\ref{app:kappa-ordering} that $\kappa_b>\kappa_c$ for all
$0<M<M_{\mathcal N}$, implying $\dot M_{\rm H}<0$.  Defining
$\Delta\kappa^2\equiv\kappa_b^2-\kappa_c^2$, the characteristic mass-loss
timescale is
\[
t_M
   \;\sim\; \frac{M}{|\dot M_{\rm H}|}
   \;\simeq\;
   \frac{48\pi\,M}{N_{\rm eff}\,\Delta\kappa^2}.
\]

In the small-black-hole regime
\[
2G_4M\ll H^{-1},\qquad H=\sqrt{\Lambda/3},
\]
one has
\[
\kappa_b\simeq \frac{1}{4G_4M},
\qquad
\kappa_c\sim \mathcal O(H).
\]

Thus, when \(\kappa_b\gg\kappa_c\),
\[
\Delta\kappa^2
\simeq \frac{1}{16G_4^2M^2},
\]
and the anomaly-induced mass-loss time is
\[
t_M\sim \frac{M}{|\dot M_{\rm H}|}
\simeq
\frac{768\pi G_4^2M^3}{N_{\rm eff}},
\]
or \(t_M\simeq 768\pi M^3/N_{\rm eff}\) in Planck units.
For a solar-mass black hole, \(M\simeq9\times10^{37}\) in Planck units,
the characteristic anomaly-induced mass-loss time is
\[
t_M\simeq  \frac{10^{117}}{N_{\rm eff}}\,t_{\rm Pl}
\sim \frac{10^{74}}{N_{\rm eff}}\,{\rm s}.
\]
This is a characteristic instantaneous mass-loss
timescale rather than an integrated lifetime.

\paragraph*{Schwinger discharge timescale.---}
The charge-loss law~\eqref{eq:dotQ-Sch} contains the Schwinger exponent
\[
\exp\!\left[-\frac{\pi m_{\rm eff}^2}{e_{\rm eff}E_b}\right],
\qquad
E_b=\frac{|Q|}{r_b^2},
\]
where \(E_b\) is the magnitude of the electric field at the black-hole
horizon.  For \(Q>0\), define
\[
\zeta\equiv \frac{\pi m_{\rm eff}^{2}r_b^2}{e_{\rm eff}Q}.
\]
In the same regime, define the dimensionless charge fraction
\[
f\equiv \frac{|Q|}{\sqrt{G_4}M}.
\]
Up to corrections of order \(H^2r_b^2\), the black-hole horizon is
approximately the Reissner--Nordstr\"om value
\[
r_b\simeq G_4M\left(1+\sqrt{1-f^2}\right),
\]
so in Planck units
\[
\zeta\simeq
\frac{\pi m_{\rm eff}^{2}}{e_{\rm eff}}\,
M\,\frac{\left(1+\sqrt{1-f^2}\right)^2}{f}.
\]
Thus the numerical size of the exponent depends on the charge fraction
\(f\).  For \(f=O(1)\), including near-extremal charge fractions, the
exponent is well below unity for stellar-mass black holes when the electron
is the dominant available charged species, up to convention-dependent
normalizations of the electromagnetic coupling.  In the small-\(f\)
approximation \(r_b\simeq2M\), this reduces to
\[
\zeta\simeq
\frac{4\pi m_{\rm eff}^{2}}{e_{\rm eff}}\,
\frac{M}{f}.
\]

In the unsuppressed regime \(\zeta\ll1\), the Schwinger law
\eqref{eq:dotQ-Sch} reduces, for \(Q>0\), to
\[
\dot Q\simeq
-\bar\beta_{\rm Sch}\ell_*
\frac{e_{\rm eff}^{2}}{2\pi r_b^2}\,Q .
\]
Thus the charge decays exponentially with e-folding time
\[
t_Q=\frac{Q}{|\dot Q|}
\simeq
\frac{2\pi r_b^2}{\bar\beta_{\rm Sch}\ell_* e_{\rm eff}^{2}}.
\]
For \(r_b\simeq2G_4M\),
\[
t_Q\simeq
\frac{8\pi G_4^2M^2}
{\bar\beta_{\rm Sch}\ell_* e_{\rm eff}^{2}}.
\]
For fixed microscopic \(\ell_*\), this scales as \(t_Q\sim M^2\).

\paragraph*{Hierarchy and physical implications.---}
By contrast, in this regime the anomaly-induced Hawking mass-loss time
scales as \(t_M\sim M^3\).
More explicitly, using \(\kappa_b\simeq(4M)^{-1}\) and
\(\kappa_c\ll\kappa_b\),
\[
t_M\simeq \frac{768\pi M^3}{N_{\rm eff}}.
\]
Combining this with the unsuppressed Schwinger estimate
\(t_Q\simeq 2\pi r_b^2/(\bar\beta_{\rm Sch}\ell_*e_{\rm eff}^2)\) in
the small-charge geometric approximation \(r_b\simeq2M\), gives
\[
\frac{t_Q}{t_M}
\simeq
\frac{N_{\rm eff}}
{96\,\bar\beta_{\rm Sch}e_{\rm eff}^{2}}\,
\frac{1}{\ell_* M}.
\]
For a microscopic matching length of order the fundamental scale,
\(\ell_*M\gg1\) for macroscopic black holes, so \(t_Q/t_M\ll1\) in the
unsuppressed-discharge regime.  Thus neutralization occurs on a timescale
parametrically short compared with
Hawking evaporation.\footnote{The Schwinger exponent depends on the charge
fraction \(f=|Q|/(\sqrt{G_4}M)\).  For an electron and a stellar-mass black hole with
\(f=O(1)\), the RN estimate is convention-dependently of order
\(10^{-6}\)--\(10^{-4}\), well below unity.  As \(|Q|\) becomes very small
the exponent eventually grows, but by then the black hole has already been
neutralized for the purposes of the macroscopic phase portrait.}

This is consistent with the worldline-instanton analysis of~\cite{Lin:2024jug}, wherein it was shown that the Schwinger exponent is
minimized in the AdS$_2\times S^2$ near-horizon region and that
contributions from larger radii are exponentially subdominant in their
asymptotically flat extremal RN setting.  This supports the assumption
that the dominant discharge rate is controlled by the near-horizon electric
field \(E_b\).  Curvature corrections, greybody factors, and matching effects may
renormalize the prefactor and the precise range of validity of the
horizon-local estimate.  Within our effective discharge model,
they enter through nonnegative rate or transmission factors, leaving the sign of \(\dot Q\) unaltered.  The rapid-discharge hierarchy
persists in the unsuppressed regime as long as the effective prefactor is
not parametrically small.

This hierarchy is the quantitative origin of the nearly vertical flow in
Fig.~\ref{fig:RNDS-sharkfin}: trajectories rapidly move toward \(q=0\) before
appreciable anomaly-driven Hawking evaporation occurs, while
the full mass evolution continues to include the electromagnetic work term
\(\Phi_b\dot Q\).  They then evolve slowly along the neutral SdS branch.
Consequently, in any theory with an operative discharge channel in the
rapid-discharge regime, controlled trajectories proceed through rapid
neutralization followed by monotonic Hawking-driven decay along the
neutral SdS channel to the empty de~Sitter endpoint.

\section{Steady states, extremal limits, and phase structure}
\label{sec:steadystates}

Thermal equilibrium of the two outer horizons requires $T_b=T_c$.  
For neutral SdS this occurs only in the Nariai limit, while in the charged
RN--dS family there is a classical lukewarm locus with $T_b=T_c\neq0$.
In addition, the static solution space contains distinguished degenerate
boundaries: the cold/extremal branch $r_-=r_b$, the charged Nariai branch
$r_b=r_c$, and their ultracold endpoint $r_-=r_b=r_c$.  None of these
classical loci is a full-flow attractor once semiclassical backreaction and
discharge are included.

\subsection{Charged Nariai limit}

One boundary of the admissible static-patch region is the charged
Nariai curve $r_b=r_c\equiv r_{\rm N}$, where both surface gravities
with respect to the original static-patch Killing time vanish.\footnote{
This refers to the Killing normalization used in the conserved-current
calculation.  Observer-normalized Nariai temperatures need not vanish and
are the appropriate quantities for local heat-capacity questions.}
Although \(\mathcal F\) vanishes identically at these
points, exact charged Nariai data are not an attracting endpoint for
controlled nondegenerate interior trajectories.  Subextremal perturbations
that split the double root move the solution into the interior of the
admissible static-patch region.  Once Schwinger discharge drives
\(|Q|\) toward zero, the geometry is funneled into a neighborhood of
the neutral SdS branch, where the proven inequality
\(\kappa_b(M,0)>\kappa_c(M,0)\) for \(0<M<M_{\mathcal N}\)
(Appendix~\ref{app:kappa-ordering}) ensures a net outward flux and
continued evaporation.

If charged fields capable of Schwinger production exist, the charged
Nariai branch is not an attracting endpoint for the physical interior
phase portrait. Formally, away from the ultracold endpoint, the reduced
adiabatic energy balance is tangent to the charged Nariai double-root
curve.  In geometric variables,
\[
\mathfrak m\equiv G_4M,\qquad
\mathfrak q\equiv \sqrt{G_4}Q,
\]
the charged Nariai double-root conditions give
\[
\mathfrak q^2
=
r_N^2\left(1-\frac{3r_N^2}{L^2}\right),
\qquad
\mathfrak m
=
r_N\left(1-\frac{2r_N^2}{L^2}\right).
\]
Hence, away from the ultracold cusp where the parametrization degenerates,
\[
\frac{d\mathfrak m}{d\mathfrak q}
=
\frac{\mathfrak q}{r_N}.
\]
On the charged Nariai curve \(\mathcal F=0\), and the work term gives
\[
\dot{\mathfrak m}
=
\frac{\mathfrak q}{r_N}\dot{\mathfrak q}
=
\frac{d\mathfrak m}{d\mathfrak q}\dot{\mathfrak q}.
\]
Thus the formal reduced flow is tangent to the double-root boundary, and
exact degenerate data require separate treatment.  The endpoint statement
used here concerns controlled nondegenerate trajectories in the
static-patch domain: subextremal perturbations that split the double root
move the solution into the interior, where the coupled anomaly-discharge
flow drives the system away from the charged equilibrium structure and
toward the neutral SdS channel.
\subsection{Extremal RN--dS limit}

The cold/extremal curve $\kappa_b=0$ forms yet another boundary of the
Reissner--Nordstr\"om--de~Sitter family.  
Classically this corresponds to a degenerate black-hole horizon, but away
from the ultracold endpoint the cosmological horizon remains simple with
\(\kappa_c>0\). 
The anomaly-driven flux therefore satisfies
$\dot M_{\rm H}\propto +\kappa_c^2>0$ along the extremal branch, so the
semiclassical evolution drives the system \emph{away} from extremality in
the mass direction.

The complementary instability in charge is governed by Schwinger discharge:
unless all charged species are absent, $\dot Q_{\rm Sch}$ has the opposite
sign of $Q$ and the magnitude $|Q|$ decreases.  
Thus the anomaly-only contribution pushes away from the cold branch in the
mass direction, while Schwinger discharge decreases $|Q|$.  Consequently, the cold branch is not an attractor of the full flow.

\paragraph*{Relation to the Cauchy horizon.---}
The outer-horizon evolution studied here is logically separate from the
inner-horizon problem of strong cosmic censorship.  Nevertheless, the two
pictures are complementary. Charged type--D geometries that possess inner horizons
are susceptible to semiclassical stress-tensor amplification at those
horizons~\cite{Poisson:1990eh,Ori:1991zz,Hollands:2019whz,Hollands:2020qpe,
Cardoso:2018nvb,Easson:2025uca}.  The present result addresses a different part of
the same problem: in the rapid-discharge regime, the outer adiabatic flow drives the
solution toward the neutral SdS channel, so it does not select a
persistent charged RN--dS endpoint with a late-time inner horizon.  Thus
the phase-space flow and the local Cauchy-horizon instability point in the
same direction, while relying on distinct semiclassical mechanisms.

\subsection{Full-flow status of the lukewarm curve}

The celebrated lukewarm curve is often interpreted as an equilibrium family because
the two outer horizons have equal temperatures.  We now show that this is
only a fixed-charge anomaly-flux condition, not an invariant condition for
the full semiclassical flow.

Let us denote the classical lukewarm locus by
\begin{equation}
\mathcal L_{\rm LW}:\qquad \kappa_b(M,Q)=\kappa_c(M,Q) .
\end{equation}
The locus is depicted in Fig.~\ref{fig:RNDS-sharkfin} by the dashed
gray curve.  On
\(\mathcal L_{\rm LW}\) the Polyakov flux \eqref{eq:fluxlaw} vanishes,
\(\mathcal F=0\).  However, the full semiclassical flow also contains the
discharge term and the associated electromagnetic work contribution, so
\(\mathcal F=0\) is not by itself a full-flow equilibrium condition.

For the standard RN--dS family, the \(Q>0\) lukewarm branch is simply
\[
\mathfrak q=\mathfrak m ,
\]
with the \(Q<0\) branch obtained by charge conjugation.  On the lukewarm
locus \(\mathcal F=0\), so the mass evolution equation gives
\[
\dot{\mathfrak m}=\varphi_b\,\dot{\mathfrak q},
\qquad
\varphi_b\equiv \frac{\mathfrak q}{r_b}.
\]
Hence
\[
\frac{d}{dv}(\mathfrak m-\mathfrak q)
=
(\varphi_b-1)\dot{\mathfrak q}.
\]
Along the subextremal lukewarm branch one has
\[
r_b>\mathfrak q=\mathfrak m,
\]
and therefore \(0<\varphi_b<1\).  Since Schwinger discharge gives
\(\dot{\mathfrak q}<0\), it follows that
\[
\frac{d}{dv}(\mathfrak m-\mathfrak q)>0.
\]
Thus the full semiclassical vector field is not tangent to the lukewarm
curve: even exactly lukewarm initial data immediately leave the lukewarm
locus once discharge is operative.

Equivalently, the lukewarm curve is only the nullcline of the Polyakov
heat flux, not a full-flow mass nullcline.  It cannot be a dynamical
attractor of the coupled semiclassical system unless all charged species
are kinematically inaccessible.  This is the only stability statement
needed for the endpoint argument; no assumption about attraction along the
fixed-charge anomaly-flow nullcline is made.
Combining these observations gives the part of the phase portrait needed
for our endpoint theorem.  The cold/extremal curve is not an attracting
invariant set, the charged Nariai boundary is not an attracting endpoint
for controlled nondegenerate interior trajectories, and the lukewarm locus
is left in the charge direction once Schwinger discharge is active.  Thus
controlled rapid-discharge trajectories waterfall toward the neutral
line (see inset in Fig. 1), where Appendix~\ref{app:kappa-ordering} guarantees monotonic mass
loss along the \(Q=0\) channel.
\section{End states and no--remnant result}
\label{sec:endstates}

We now formulate the \emph{no-remnant consequence} of the phase portrait described
above.  The result applies to controlled nondegenerate trajectories in the
static-patch domain whenever an operative Schwinger discharge channel
exists.  In this rapid-discharge regime, the semiclassical flow carries
charged data toward the neutral SdS channel, where
\(\kappa_b(M,0)>\kappa_c(M,0)\) enforces continued mass loss.  Thus the
classical distinguished loci of the RN--dS family cannot be selected as
late-time remnants by the coupled anomaly-discharge evolution.  The limit
\(M\to0\) should be understood as the formal endpoint of the adiabatic
semiclassical evolution; possible Planck-scale corrections near
\(M\sim M_{\rm Pl}\) lie outside the regime of control.

\begin{theorem}[Global semiclassical endpoint]
\label{thm:endstate}
Let \((M_0,Q_0)\) be sub-Nariai initial data whose adiabatic trajectory
remains in the RN--dS static-patch domain during the rapid-discharge
stage.  Suppose that at least one charged species provides an operative
Schwinger discharge channel with \(t_Q\ll t_M\) until the charge is
macroscopically negligible.  Within the anomaly-induced semiclassical
description governed by Eqs.~\eqref{eq:dyn-mu}--\eqref{eq:dyn-q}, the
controlled nondegenerate flow selects the formal endpoint
\[
(M,Q)=(0,0).
\]
Thus the late-time endpoint in this approximation is empty de~Sitter
space, rather than an extremal, lukewarm, or charged/neutral Nariai
remnant.
\end{theorem}

\begin{proof}
For \(Q\neq0\), the Schwinger term satisfies
\[
\mathrm{sgn}(Q)\dot Q_{\rm Sch}<0,
\]
so \(|Q|\) decreases monotonically.  In the unsuppressed regime the
discharge time satisfies \(t_Q\ll t_M\) for macroscopic black holes, so
controlled trajectories are rapidly driven into an effectively neutral
neighborhood of \(Q=0\), up to the exponentially slow terminal Schwinger
tail.

Once the charge is macroscopically negligible, the evolution is governed
by the neutral SdS channel.  Appendix~\ref{app:kappa-ordering} proves
\[
\kappa_b(M,0)>\kappa_c(M,0),
\qquad 0<M<M_{\mathcal N},
\]
and hence
\[
\dot M_{\rm H}
=
-\frac{N_{\rm eff}}{48\pi}
\bigl[\kappa_b^2(M,0)-\kappa_c^2(M,0)\bigr]
<0 .
\]
Thus the mass decreases monotonically along the neutral branch toward the
formal adiabatic endpoint \(M\to0\).

The cold/extremal curve, the charged Nariai boundary, and the lukewarm curve
are not attracting endpoints for controlled nondegenerate trajectories, as
shown in Sec.~\ref{sec:steadystates}.  
Therefore, within the formal adiabatic continuum description, the
controlled semiclassical endpoint is
\[
M(v)\to0,\qquad Q(v)\to0,
\]
corresponding to empty de~Sitter space.
\end{proof}

\begin{figure}[t]
    \centering
    \includegraphics[width=.9\columnwidth]{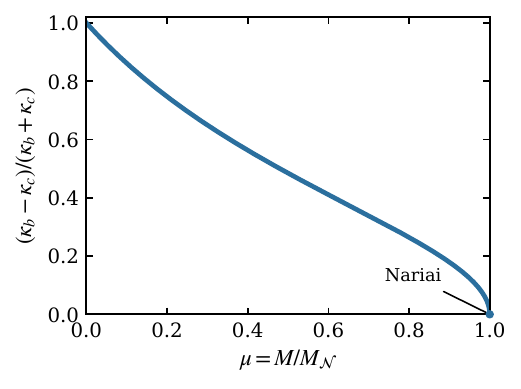}
    \caption{
Relative surface-gravity difference along the neutral Schwarzschild--de~Sitter
branch, $(\kappa_b-\kappa_c)/(\kappa_b+\kappa_c)$, plotted as a function of
$\mu \equiv M/M_{\mathcal N}$.  The curve remains strictly positive for all
$0<\mu<1$, showing that the black-hole horizon has larger Killing temperature than
the cosmological horizon throughout the sub-Nariai neutral family.  The relative
difference vanishes only in the Nariai limit, $\mu=1$, where
$\kappa_b=\kappa_c=0$.  This numerically illustrates the analytic ordering
$\kappa_b>\kappa_c$ proved in Appendix~\ref{app:kappa-ordering}, which implies
$\kappa_b^2-\kappa_c^2>0$ and hence monotonic anomaly-driven mass loss along
the neutral SdS channel.
}
    \label{fig:kappa-diff-SdS}
\end{figure}

\begin{corollary}[Absence of semiclassical remnants]
No finite-mass,
charged, or neutral remnants form under the assumptions of Theorem~\ref{thm:endstate}.  
In particular, extremal RN--dS and charged or neutral Nariai geometries cannot be semiclassical endpoints for controlled nondegenerate rapid-discharge
trajectories.
\end{corollary}

\section{Information-theoretic interpretation}
\label{sec:info}

The same two-dimensional semiclassical framework provides the
natural input for an island estimate of the radiation entropy.  Here, we do not
attempt a microscopic derivation of the exact fine-grained Page curve
for RN--dS evaporation.  Such formidable derivation would require specifying the
radiation region, the global quantum state, the UV subtraction convention
in $S_{\rm out}$, and the fully time-dependent QES geometry.  We do, however, extract the robust semiclassical consequence of the dynamical solution:
the background entering the island problem is driven rapidly toward the
neutral SdS channel whenever Schwinger discharge is operative.

Island constructions in cosmological and de~Sitter settings have been
studied in several complementary models
\cite{Hartman:2020khs,Aalsma:2021bit,Shaghoulian:2022fop,Geng:2021wcq};
we take these as motivation for a conservative saddle-competition
estimate on the dynamically selected RN--dS background.

The generalized entropy of an island endpoint is
\[
S_{\rm gen}[\partial\mathcal I]
 =
 \frac{A(\partial\mathcal I)}{4G_4}
 +S_{\rm out}(\mathcal R\cup\mathcal I),
\]
or equivalently, in the two-dimensional reduction, $A=4\pi X$.  The fine-grained radiation entropy is estimated by the usual
competition between the no-island and island saddles~\cite{Almheiri:2019psf,Penington:2019npb,Almheiri:2020cfm},
\[
S_{\rm fine}(v)
 =
 \min \{S_{\rm no\,island}(v),S_{\rm island}(v) \}.
\]
In the adiabatic Unruh--dS state, the no-island branch is controlled by
the accumulated Hawking entropy in the static patch, while the island
branch is dominated by the area term evaluated at a QES exponentially
close to the black-hole horizon,
\[
S_{\rm island}(v)
 \simeq
 \frac{A_b(v)}{4G_4}+O(N_{\rm eff})
 =
 \frac{\pi r_b(v)^2}{G_4}+O(N_{\rm eff}).
\]
The Page time is therefore estimated by the branch-crossing condition
\[
S_{\rm no\,island}(v_{\rm Page})
 \simeq
S_{\rm island}(v_{\rm Page}).
\]

Within the present adiabatic approximation, charge affects this estimate
through the time-dependent background
$r_b(M(v),Q(v))$, $\kappa_b(M(v),Q(v))$, and
$\kappa_c(M(v),Q(v))$.  In the rapid-discharge regime studied above,
$t_Q\ll t_M$, so $Q(v)$ becomes negligible long before the
mass changes appreciably.  The subsequent island problem is therefore a
smooth deformation of the neutral SdS problem and quickly reduces to the
neutral channel.  Thus the phase portrait not only fixes the
semiclassical endpoint but also identifies the background on which the
late-time QES competition is to be evaluated.

Since the Page estimate is controlled by the slow Hawking evolution, while
the discharge time satisfies $t_Q\ll t_M$ in the macroscopic
unsuppressed-discharge regime, charge-dependent corrections to the
late-time island saddle are parametrically transient.  The geometry has
already become effectively neutral by the time the no-island and island
branches compete.  Thus the charged problem does not require a new
late-time QES mechanism: the relevant saddle competition reduces to the
neutral SdS channel, with only early-time charge-dependent corrections
during the discharge stage.

Should discharge be kinematically forbidden, the same formalism results instead in
a fixed-charge anomaly-flow problem.  In that nongeneric regime the
geometry may approach the anomaly-flux equilibrium locus rather than the
neutral evaporation channel, and the corresponding entropy history is appropriately analyzed on that fixed-charge background rather than inferred
from the rapid-discharge case.  A full calculation of
the fixed-charge QES geometry lies beyond the scope of this work.

\section{Conclusions}

We have developed an analytic semiclassical model for the nonequilibrium
discharge and evaporation of Reissner--Nordstr\"om--de~Sitter black holes
within the spherically symmetric, anomaly-induced adiabatic approximation.
The Polyakov sector fixes a conserved horizon-to-horizon Killing-energy
flux proportional to \(\kappa_b^2-\kappa_c^2\), while the full mass
evolution includes the electromagnetic work term associated with charge
loss.  Together with a near-horizon Schwinger discharge law, this gives a
closed evolution system for the mass and charge of the black hole.

The main result is a two-stage endpoint paradigm.  In the rapid-discharge
regime, any active sufficiently light charged species drives the system
toward the neutral line on a timescale parametrically shorter than the
anomaly-induced Hawking mass-loss time.  Once the trajectory reaches an effectively neutral neighborhood, the
neutral SdS ordering proved in Appendix~\ref{app:kappa-ordering} fixes the
sign of the remaining Hawking evolution and implies monotonic mass loss.  Controlled rapid-discharge trajectories therefore select the
formal adiabatic endpoint \((M,Q)=(0,0)\), corresponding to empty
de~Sitter space.

This endpoint does not follow by assuming a global temperature ordering in
the charged RN--dS family.  The charged solution space contains the
lukewarm locus \(T_b=T_c\), where the anomaly-induced heat flux vanishes.
The crucial point is dynamical: the lukewarm curve is only a nullcline of
the Polyakov flux, not an invariant locus or mass nullcline of the full
anomaly-discharge flow.  Similarly, the cold/extremal branch, charged
Nariai branch, and ultracold point are distinguished classical loci, but
they are not attracting endpoints when charge loss is dynamically active.
The nonequilibrium dynamics reorganizes the classical phase diagram:
charge is shed first, and the neutral SdS evaporation channel determines
the endpoint.

The same framework gives a controlled thermodynamic interpretation.  The anomaly-mediated heat flow satisfies the coarse-grained horizon
generalized second law in the static patch, while the time-dependent background supplies the appropriate
semiclassical geometry for island/QES estimates.  In the
rapid-discharge regime, the charge-dependent part of the geometry is
transient on the timescale relevant for late-time saddle competition, so
the information-theoretic estimate reduces smoothly to the neutral SdS
channel.

Our framework has a precise domain of validity: spherical
symmetry, the Polyakov anomaly sector for conformal s-wave matter, an
adiabatic Unruh--de~Sitter state, and a near-horizon effective description
of Schwinger discharge.  Greybody factors, dilaton-dependent anomaly terms,
higher-derivative corrections, and finite-\(N\) effects can modify
coefficients and state-dependent details.  However, within any regime in
which the neutral flux retains the sign fixed by
\(\kappa_b>\kappa_c\) and discharge remains faster than mass loss, these
corrections are not expected to change the endpoint mechanism.  Naturally, Planck-scale physics
near the formal endpoint lies outside the regime of control.

More broadly, the RN--dS system suggests a general endpoint principle for
charged multi-horizon spacetimes: when a conserved horizon-to-horizon
energy current is combined with a fast discharge channel, charged
equilibrium loci need not control the late-time fate.  Instead, the
neutral surface-gravity ordering can act as the endpoint selector.  Testing
this principle beyond spherical symmetry, including slow rotation,
greybody-resolved fluxes, and more general anomaly sectors, is a fruitful
direction for future work.

\acknowledgments
It is a pleasure to thank Lars Aalsma, Paul Davies and Tanmay Vachaspati for useful discussions. This work is supported by the U.S. Department of Energy, Office of High Energy Physics, under Award Number DE-SC0019470.

\appendix
\section{Ordering of Horizon Surface Gravities in Sub-Nariai SdS}
\label{app:kappa-ordering}

This appendix provides the mathematical lemma used in Sec.~\ref{sec:endstates}: for every
neutral Schwarzschild--de~Sitter geometry with $0<M<M_{\mathcal N}$, the black-hole horizon has larger Killing surface gravity, equivalently
larger Killing temperature in the normalization used here,
\(\kappa_b>\kappa_c\). 
The inequality fixes the sign of the anomaly-induced mass change
\(\dot M_{\rm H}=-\mathcal F\) along the neutral line of the \((M,Q)\)
phase portrait and thereby guarantees monotonic mass loss toward
\(M\to0\).
We now provide a brief and rigorous proof of the inequality.

\subsection{Setup and notation}

For \(0 < M < M_{\mathcal N}\), the metric function
\[
  f(r) = 1 - \frac{2G_4M}{r} - \frac{\Lambda r^2}{3}
\]
admits three real roots,
\[
  r_0 < 0 < r_b < r_c,
\]
corresponding to the unphysical negative root \(r_0\), the black-hole
horizon \(r_b\), and the cosmological horizon \(r_c\).

In Sec.~\ref{sec:2d-theory} we showed that the surface gravities can be written in the
closed forms
\begin{align}
  \kappa_b &= 
  \frac{(r_c - r_b)\,(2 r_b + r_c)}
       {2\,r_b\,(r_b^2 + r_b r_c + r_c^2)},
  \label{app:kappab}
  \\
  \kappa_c &= 
  \frac{(r_c - r_b)\,(r_b + 2 r_c)}
       {2\,r_c\,(r_b^2 + r_b r_c + r_c^2)}.
  \label{app:kappac}
\end{align}
These expressions follow directly from the factorization
\[
  f(r) = -\frac{\Lambda}{3r}
         (r - r_b)(r - r_c)(r - r_0)
\]
together with \(\kappa_i = \tfrac{1}{2}|f'(r_i)|\), and may be taken as equivalent to the corresponding surface-gravity formulas in the main text.

\subsection{Proof of \(\kappa_b > \kappa_c\)}

In the sub-Nariai regime the following quantities are manifestly positive:
\[
  r_c - r_b > 0, \qquad
  r_b > 0, \qquad
  r_c > 0, \qquad
  r_b^2 + r_b r_c + r_c^2 > 0.
\]
Thus both \(\kappa_b\) and \(\kappa_c\) are strictly positive.

To compare them, factor out the positive constant
\[
  C \equiv \frac{r_c - r_b}
                {2\,(r_b^2 + r_b r_c + r_c^2)} > 0.
\]
Then \eqref{app:kappab}--\eqref{app:kappac} take the form
\begin{align}
  \kappa_b &= C\,\frac{2 r_b + r_c}{r_b},
  \\
  \kappa_c &= C\,\frac{r_b + 2 r_c}{r_c}.
\end{align}
Since \(C>0\), the inequality \(\kappa_b > \kappa_c\) is equivalent to
\begin{equation}
  \frac{2 r_b + r_c}{r_b}
    \;>\;
  \frac{r_b + 2 r_c}{r_c}.
  \label{ineq:core}
\end{equation}

Rewrite both sides:
\[
  \frac{2 r_b + r_c}{r_b} = 2 + \frac{r_c}{r_b}, \qquad
  \frac{r_b + 2 r_c}{r_c} = 2 + \frac{r_b}{r_c}.
\]
Subtracting \(2\) from \eqref{ineq:core} yields the equivalent inequality
\[
  \frac{r_c}{r_b} > \frac{r_b}{r_c}.
\]
Let \(x \equiv r_c / r_b\).  
In the sub-Nariai region one has \(r_c > r_b > 0\), hence \(x>1\).  
The inequality becomes
\[
  x > \frac{1}{x}.
\]
Multiplying by \(x>0\),
\[
  x^2 > 1,
\]
which is clearly true for all \(x>1\).  
This establishes the result:
\[
  \kappa_b > \kappa_c > 0
  \qquad
  \text{for all}\quad 0 < r_b < r_c.
\]

\subsection{Nariai limit}

As \(r_b \to r_c\), the common prefactor \(C\) tends to zero and both
surface gravities vanish:
\[
  \kappa_b \to 0,
  \qquad
  \kappa_c \to 0.
\]
Thus the Nariai spacetime is the unique double-zero of the surface
gravities.  It is a zero-flux boundary configuration, but not a
stable endpoint: perturbations that split the horizons restore the
neutral ordering $\kappa_b>\kappa_c$ and restart evaporation.

\medskip\noindent
\subsection{Implication for semiclassical evolution}

\noindent
Since $\kappa_b>\kappa_c$ for all $0<M<M_{\mathcal N}$, the anomaly
flux satisfies \(\mathcal F(M,0)>0\), equivalently
\(\dot M_{\rm H}(M,0)=-\mathcal F(M,0)<0\), on the entire interval, so the
neutral mass decreases monotonically to $M\to0$.
The Nariai point ($\kappa_b=\kappa_c=0$) is a static solution at the upper mass bound.  
Although radiation from the cosmological horizon provides an external influx, it is always subdominant to the black-hole outflow for $M<M_{\mathcal N}$, so the Nariai configuration is not approached dynamically from the
sub-Nariai neutral branch by Hawking evaporation.

\section{Conserved Killing current and Polyakov flux}
\label{app:flux-kappa-diff}

This appendix verifies two ingredients used in the main text in the presence of charge:
(i) the existence of a radially conserved Killing-energy current in the static patch, and
(ii) the resulting appearance of the combination $\kappa_b^2-\kappa_c^2$ in the
\emph{Polyakov} (pure conformal) contribution to the net horizon-to-horizon flux.
The derivation is the same as in the neutral SdS case, by simply replacing $\xi(r)$ by $\xi(r;M,Q)$.

\subsection{Static patch and a conserved flux}

In the static region $r_b<r<r_c$, the four-dimensional RN--dS metric is
\begin{equation}
ds^2=-\xi(r)\,dt^2+\xi(r)^{-1}dr^2+r^2d\Omega_2^2,
\qquad
\xi(r)=1-\frac{2G_4M}{r}+\frac{G_4Q^2}{r^2}-\frac{\Lambda r^2}{3},
\label{app:RNdS_metric}
\end{equation}
with $\xi(r_b)=\xi(r_c)=0$ and $0<r_b<r_c$.
Introduce the tortoise coordinate $dr_\ast/dr=\xi^{-1}$ and null coordinates
\begin{equation}
u=t-r_\ast,\qquad v=t+r_\ast,
\label{app:nullcoords}
\end{equation}
so the reduced two-dimensional line element takes the conformal form
\begin{equation}
ds_2^2=-\xi(r)\,du\,dv.
\label{app:2dmetric_uv}
\end{equation}

Let \(T_{ab}\) be a symmetric, stationary, covariantly conserved
two-dimensional stress tensor
\begin{equation}
\nabla_a T^{ab}=0,
\label{app:stress_cons}
\end{equation}
in particular the anomaly-induced expectation value obtained from the Polyakov action.
The static Killing vector $\zeta^a=(\partial_t)^a$ defines the Killing-energy current
\begin{equation}
J^a \equiv -T^{a}{}_{b}\,\zeta^{b}.
\label{app:Killing_current_def}
\end{equation}
Using $\nabla_{(a}\zeta_{b)}=0$ and \eqref{app:stress_cons},
\begin{equation}
\nabla_a J^a
=-(\nabla_a T^{a}{}_{b})\zeta^b - T^{ab}\nabla_a\zeta_b
=0.
\label{app:Killing_current_cons}
\end{equation}
In the $(t,r)$ coordinates of \eqref{app:RNdS_metric}, $\sqrt{-g_{(2)}}=1$, so \eqref{app:Killing_current_cons} implies
\begin{equation}
\partial_r J^r=0
\qquad\Rightarrow\qquad
\partial_r\big(T^{r}{}_{t}\big)=0,
\label{app:flux_const_trt}
\end{equation}
since $J^r=-T^{r}{}_{t}$ for $\zeta=\partial_t$.
Thus the mixed component $T^{r}{}_{t}$ is \emph{radially constant} throughout the static patch.

In the null coordinates \eqref{app:nullcoords} one has the identity
\begin{equation}
T^{r}{}_{t}=T_{vv}-T_{uu},
\label{app:Trt_null}
\end{equation}
so we define the conserved mixed component
\begin{equation}
\mathcal I\;\equiv\;T^{r}{}_{t}\;=\;T_{vv}-T_{uu},
\qquad\Rightarrow\qquad
\partial_r\mathcal I=0.
\label{app:J_def}
\end{equation}
This $\mathcal I$ is the signed mixed stress-tensor component.  The
positive outward flux used in the main text is defined with the opposite
sign when the black-hole horizon is hotter: $\mathcal F_{\rm P}\equiv-\mathcal I$.

\subsection{Horizon regularity and universal Schwarzian offsets}

Let $r_h$ be any \emph{simple} Killing horizon of $\xi(r)$, with surface gravity
\begin{equation}
\kappa_h \equiv \tfrac12\,|\xi'(r_h)|.
\label{app:kappa_def}
\end{equation}
Define future-horizon affine Kruskal coordinates separately in each chiral sector:
\begin{equation}
U_h=-e^{-\kappa_h u}\quad\text{(outgoing sector)},
\qquad
V_h=-e^{-\kappa_h v}\quad\text{(ingoing sector)}.
\label{app:Kruskal_affine}
\end{equation}
For the pure conformal (Polyakov) sector with central charge $c=N$, the chiral components obey
the projective (Schwarzian) transformation law
\begin{align}
\langle T_{uu}\rangle
&=\Big(\frac{dU_h}{du}\Big)^{\!2}\langle T_{U_hU_h}\rangle
-\frac{N}{24\pi}\{U_h,u\},
\label{app:proj_u}
\\
\langle T_{vv}\rangle
&=\Big(\frac{dV_h}{dv}\Big)^{\!2}\langle T_{V_hV_h}\rangle
-\frac{N}{24\pi}\{V_h,v\},
\label{app:proj_v}
\end{align}
where $\{\cdot,\cdot\}$ is the Schwarzian derivative.
For the exponential maps in \eqref{app:Kruskal_affine} one has
\begin{equation}
\{U_h,u\}=\{V_h,v\}=-\frac12\,\kappa_h^2.
\label{app:Schwarzian_exp}
\end{equation}

In the static Polyakov solution the chiral components may be written as
\[
\langle T_{uu}\rangle=B(r)+t_u,\qquad
\langle T_{vv}\rangle=B(r)+t_v,
\]
where \(B(r)\) is the common state-independent geometric term and
\(t_u,t_v\) are state-dependent chiral constants.  With the standard
Polyakov normalization,
\begin{equation}
B(r_h)=-\Theta_h,\qquad
\Theta_h\equiv \frac{N}{48\pi}\kappa_h^2.
\label{app:Theta_def}
\end{equation}
Near a simple horizon, an affine Kruskal coordinate in the outgoing sector
may be taken as
\(U_h=-e^{-\kappa_h u}\).  Since
\[
\langle T_{U_hU_h}\rangle
=
\left(\frac{du}{dU_h}\right)^2
\langle T_{uu}\rangle ,
\]
finiteness of \(\langle T_{U_hU_h}\rangle\) at \(U_h=0\) requires
\[
\langle T_{uu}\rangle=O(U_h^2).
\]
Thus the leading horizon value must vanish,
\[
B(r_h)+t_u=0,
\]
and outgoing regularity at \(r_h\) fixes
\[
t_u=\Theta_h.
\]
Similarly, ingoing regularity in an affine coordinate
\(V_h=-e^{-\kappa_h v}\)
fixes
\[
t_v=\Theta_h.
\]

\subsection{Two-horizon outer prescription and the $(\kappa_b^2-\kappa_c^2)$ combination}

In the RN--dS static patch the relevant outer horizons are the black-hole horizon $r_b$ and the
cosmological horizon $r_c$. The ``two-horizon'' outer prescription used in the main text is:
\begin{itemize}
\item impose \emph{outgoing} affine regularity on the future black-hole horizon $\mathcal{H}_b^{+}$,
so \(t_u=\Theta_b\);
\item impose \emph{ingoing} affine regularity on the future cosmological horizon $\mathcal{H}_c^{+}$,
so \(t_v=\Theta_c\).
\end{itemize}
Equivalently,
\begin{equation}
t_u=\Theta_b=\frac{N}{48\pi}\kappa_b^2,
\qquad
t_v=\Theta_c=\frac{N}{48\pi}\kappa_c^2,
\label{app:two_horizon_bc}
\end{equation}
where $\kappa_{b,c}=\tfrac12|\xi'(r_{b,c})|$ are defined by \eqref{app:kappa_def} for the charged
$\xi(r;M,Q)$.

Since the common term \(B(r)\) cancels,
\begin{equation}
\mathcal I=T_{vv}-T_{uu}=t_v-t_u
\end{equation}
is radially conserved.  Using Eq.~\eqref{app:two_horizon_bc}, one obtains
\begin{align}
\mathcal I
&=\Theta_c-\Theta_b
=-\frac{N}{48\pi}\left(\kappa_b^2-\kappa_c^2\right),\\
\mathcal F_{\rm P}
&\equiv-\mathcal I
=\frac{N}{48\pi}\left(\kappa_b^2-\kappa_c^2\right).
\label{app:flux_kappa_diff}
\end{align}

\subsection{Adiabatic evolution and first-law energy balance with discharge}

Equation \eqref{app:flux_kappa_diff} is derived for a fixed static patch.
In the large-$N$/adiabatic regime used in the main text, the evolution is quasistatic: at each
advanced time $v$ the geometry is well-approximated by a static RN--dS patch with parameters
$(M(v),Q(v))$, so \eqref{app:flux_kappa_diff} applies instantaneously with
$\kappa_{b,c}=\kappa_{b,c}(M(v),Q(v))$.

When $Q(v)$ changes due to Schwinger discharge, the \emph{total} mass evolution includes the
electromagnetic work term. The Polyakov sector fixes the outward-oriented signed flux.  In the pure
two-dimensional theory this is $\mathcal F_{\rm P}$ in \eqref{app:flux_kappa_diff}; in the
four-dimensional matched evolution we replace $N$ by $N_{\rm eff}$ and denote the result by
$\mathcal F$.  The first-law energy balance then yields
\begin{equation}
\dot M = -\,\mathcal F + \Phi_b\,\dot Q,
\qquad
\Phi_b=\frac{Q}{r_b},
\label{app:Mdot_energy_balance}
\end{equation}
as used in Eq.~\eqref{eq:Mdot} of the main text.

\subsection{Remark on additional anomaly terms}

If dilaton-dependent anomaly terms are included, as can occur for
spherically reduced four-dimensional matter, the horizon regularity
offsets need not be the pure Polyakov values
\(\Theta_h=N\kappa_h^2/(48\pi)\).  They can be shifted by additional
horizon-local terms depending on the dilaton and its derivatives.  The
conserved-current argument nevertheless remains valid: for any
covariantly conserved, stationary semiclassical stress tensor on the
static patch, \(T^r{}_t=T_{vv}-T_{uu}\) is radially constant.  Imposing
the same two-horizon regularity prescription therefore gives a
difference-of-offsets structure
\[
\mathcal I=\Theta^{\rm tot}_c-\Theta^{\rm tot}_b,
\qquad
\mathcal F=\Theta^{\rm tot}_b-\Theta^{\rm tot}_c .
\]
In the pure Polyakov sector,
\[
\Theta^{\rm tot}_h=\Theta_h=\frac{N}{48\pi}\kappa_h^2,
\]
and this reduces to Eq.~\eqref{app:flux_kappa_diff}.


\bibliography{RNdS_references}

\end{document}